\documentclass[sigconf,nonacm]{acmart}

\usepackage{algorithm}
\usepackage{algpseudocode}
\usepackage{balance}
\usepackage{tikz}
\usetikzlibrary{arrows.meta,positioning,patterns,calc}
\makeatletter
\@ACM@balancefalse
\makeatother
\algrenewcommand\algorithmicrequire{\textbf{Input:}}
\algrenewcommand\algorithmicensure{\textbf{Output:}}

\setcopyright{none}
\title[Structure-Aware SAT--CEGAR for HCP]{\textsc{StructHCP}: Structure-Aware SAT--CEGAR for the Hamiltonian Cycle Problem}

\author{Quan Van Vi}
\affiliation{
  \institution{Vietnam National University, Hanoi} \city{Hanoi} \country{Vietnam}} \email{22021149@vnu.edu.vn}

\author{Tuyen Van Kieu}
\affiliation{
  \institution{Vietnam National University, Hanoi} \city{Hanoi} \country{Vietnam}} \email{tuyenkv@vnu.edu.vn}

\author{Khanh Van To}
\affiliation{
  \institution{Vietnam National University, Hanoi} \city{Hanoi} \country{Vietnam}} \email{khanhtv@vnu.edu.vn}

\newcommand{\BibTeX}{\rm B\kern-.05em{\sc i\kern-.025em b}\kern-.08em\TeX}

\newcommand{\method}{\textsc{StructHCP}}

\begin{abstract}
Many exact Hamiltonian-cycle solvers use the same edge-level representation for every input, even when graph structure admits a smaller set of global decisions. We present \method, a SAT--CEGAR pipeline that selects a representation through inexpensive structural recognizers. Specialized searches handle high-degree clusters and paired chains, while complete directed and undirected encodings handle cubic and general graphs. Every returned cycle is reconstructed when necessary and verified in the input graph. On the 1,001-instance FHCP Challenge Set, \method{} solves 826 instances, compared with 741 when the same general SAT--CEGAR engine is applied directly to every unreduced input graph. Relative to this ablation, \method{} reduces PAR-2 by 32.9\%. Formulation selection solves 82 more instances than the contracted general encoding; high-degree-cluster and cubic graphs account for 81 of this difference. \method{} also solves 34 of the 51 instances solved by neither HCup nor Cut-set CEGAR, raising the combined coverage of the three solvers from 950 to 984. Twelve of its solutions lie outside the union of published per-instance SAT results analyzed here, and six also lie outside the published exact result sets in our comparison. These results show that selecting the representation can complement connectivity refinement on structured HCP instances.
\end{abstract}


\begin{document}
\hypersetup{
  pdftitle={StructHCP: Structure-Aware SAT--CEGAR for the Hamiltonian Cycle Problem},
  pdfauthor={Quan Van Vi; Tuyen Van Kieu; Khanh Van To},
  pdfsubject={Exact SAT solving for the Hamiltonian cycle problem},
  pdfkeywords={Hamiltonian cycle, SAT, CEGAR, algorithm selection}
}
\raggedbottom
\maketitle

\section{Introduction}
\label{sec:introduction}

The Hamiltonian cycle problem (HCP) asks whether an undirected graph contains a simple cycle that visits every vertex exactly once. It is a classical NP-complete problem~\cite{karp1972} and a basic model for exact graph search. HCP also exposes a central difficulty in combinatorial reasoning: local constraints are easy to state, whereas global connectivity is not. In the usual edge-based SAT encoding, exactly two incident edges are selected at every vertex. These constraints describe a cycle cover, which may contain several disjoint cycles rather than one Hamiltonian cycle.

Counterexample-guided abstraction refinement (CEGAR) was introduced for model checking~\cite{clarke2000cegar}. Soh et al. later applied this idea to HCP: their 2014 incremental SAT method starts with degree constraints and adds connectivity clauses when the current solution contains disconnected cycles~\cite{soh2014incremental}. Ohashi et al. subsequently strengthened the refinement step with cut-set constraints, while HCup combines cut-set clauses, cycle merging, and user propagators~\cite{ohashi2025cutset,fazekas2024userpropagators,funder2026hcup}. These techniques improve connectivity reasoning while encoding the input graph directly. A single representation can therefore expose very different search spaces across graph families. Degree-two paths contain forced traversals, cubic graphs admit directed choices, and clustered graphs may be governed by a much smaller set of connections between clusters.

This variation suggests treating the representation as part of the solving decision. Structural decomposition and vertex elimination already show how graph transformations can preserve Hamiltonicity~\cite{coudert2017defi,zhou2024vertex}. More generally, algorithm selection uses inexpensive instance features to exploit complementary procedures, and encoding selection has been studied for directed Hamiltonian cycles in ASP~\cite{kerschke2019algorithmselection,liu2019encodingselection}. The unresolved question is whether one exact SAT pipeline can select structurally appropriate formulations while retaining a complete decision path. It must also preserve enough information to reconstruct every cycle in the input graph.

We address this question with \method{}, a structure-aware SAT--CEGAR solver. The solver applies inexpensive recognizers in a fixed order. Compact searches handle high-degree clusters and paired chains, a directed encoding handles cubic graphs, and a general undirected encoding handles all remaining inputs. The general branch first contracts maximal degree-two paths and records each path behind a required virtual edge. Candidate-producing branches fall back to the complete general branch when they do not find a cycle. Every returned cycle is reconstructed when necessary and verified in the input graph.

The four branches expose different global decisions to the SAT solver. The high-degree formulation selects connections between clusters rather than all internal edges. The paired-chain formulation selects orientations and transitions between contracted pairs. The cubic formulation chooses one predecessor and one successor per vertex, while the general formulation retains undirected edge variables after forced paths have been removed. The principal contribution is an exact solver architecture in which graph structure determines the representation before connectivity refinement begins.

We make three contributions. First, we integrate these four representations into one exact pipeline with complete cubic and general branches. Second, we formalize degree-two path contraction through required virtual edges and an explicit reconstruction map, and establish end-to-end correctness. The same framework supports separation-pair reduction as an optional extension whose empirical contribution is evaluated separately. Third, we provide a family-level ablation and complementary-coverage analysis on the 1,001-instance FHCP Challenge Set. \method{} solves 826 instances, compared with 741 when the same general SAT--CEGAR engine is used without formulation selection or degree-two contraction. It also solves 34 instances outside the union of two recent SAT-based reruns. A separate labelled cubic census evaluates both Hamiltonian-cycle construction and UNSAT conclusions.

\section{Background and Related Work}
\label{sec:background}

\subsection{HCP and SAT--CEGAR}
\label{sec:preliminaries}

Let $G=(V,E)$ be a finite undirected simple graph, with $n=|V|\geq3$. For $S\subseteq V$, let $\delta_G(S)$ be the cut containing the edges with exactly one endpoint in $S$, and write $\delta_G(v)=\delta_G(\{v\})$ for the edges incident to vertex $v$. Each edge $e\in E$ has a Boolean variable $x_e$, where $x_e=1$ means that the candidate cycle uses $e$. The degree formula
\begin{equation}
  \Phi_{\mathrm{deg}}(G)=
  \bigwedge_{v\in V}
  \left(\sum_{e\in\delta_G(v)}x_e=2\right)
  \label{eq:degree}
\end{equation}
is a Boolean cardinality constraint: for each vertex, exactly two incident edge variables must be true. Consequently, every satisfying assignment describes a collection of vertex-disjoint cycles covering $V$. It represents a Hamiltonian cycle exactly when this collection has one component.

For every nonempty proper subset $S\subset V$, a Hamiltonian cycle must enter and leave $S$. At least two selected edges therefore cross the cut:
\begin{equation}
  \sum_{e\in\delta_G(S)}x_e\geq2.
  \label{eq:cutset}
\end{equation}
Under Equation~\eqref{eq:degree}, every selected cut has even cardinality. The Boolean clause $\bigvee_{e\in\delta_G(S)}x_e$ requires at least one selected crossing edge; parity then makes this equivalent to the lower bound in Equation~\eqref{eq:cutset}. Such cut constraints go back to the classical subtour-elimination formulation of Dantzig, Fulkerson, and Johnson~\cite{dantzig1954tsp}.

For a proper cycle $C$, let $E(C)$ denote its selected edges. The subcycle-blocking clause $\bigvee_{e\in E(C)}\neg x_e$ requires the next model to omit at least one edge of $C$. CEGAR starts from the degree formula, decodes each model into cycles, and adds valid connectivity clauses for every proper cycle. It stops when the model is one spanning cycle or the refined formula becomes unsatisfiable.

For a vertex set $M\subseteq V$ and distinct endpoints $s,t\in M$, a spanning $s$--$t$ path uses one selected edge at $s$ and $t$ and two at every vertex in $M\setminus\{s,t\}$. This degree relaxation may still contain the endpoint path together with disjoint cycles, so the same refinement principle removes those cycles. The structural solvers use this formulation as their local path primitive.

\subsection{Exact HCP Solving}
\label{sec:related-approaches}

\paragraph{SAT encodings and refinement.}
Exact SAT methods differ in when they enforce the condition that the selected edges form one cycle. Static encodings express connectivity before search through arithmetic or ordering constraints~\cite{heule2021chinese,zhou2020encoding}. Soh et al. introduced SAT--CEGAR for HCP in 2014~\cite{soh2014incremental}. Their method starts with degree constraints and adds blocking clauses whenever the current solution contains disconnected cycles. Ohashi et al. later strengthened this refinement step by replacing subcycle-blocking clauses with cut-set constraints~\cite{ohashi2025cutset}. HCup supports complete-model refinement and user propagators that detect or prevent subcycles during search~\cite{fazekas2024userpropagators,funder2026hcup}. These methods encode the input graph directly in their main SAT formulations.

\paragraph{Structural reductions.}
Triconnected decomposition identifies separation pairs, and SPQR trees organize the resulting components~\cite{hopcroft1973triconnected,dibattista1996spqr}. Cohen and Coudert use such pairs to split an FHCP graph into Hamiltonian $u$--$v$ path subproblems and join the paths afterwards~\cite{coudert2017defi}. Zhou's vertex-elimination encoding relates a cycle in a reduced graph to one in the original graph through additional SAT constraints~\cite{zhou2024vertex}. \method{} contracts forced degree-two paths; we evaluate separation-pair reduction as an optional variant.

Structural constraints have also been integrated directly into TSP models~\cite{isoart2019structural}, while vertex-elimination encodings provide alternative reachability and acyclicity representations~\cite{rankooh2022elimination,zhou2024vertex}. Earlier HCP experiments likewise showed that variable dependencies can make one SAT model unsuitable across instances and motivated alternative models~\cite{prestwich2003chains}. Our reductions differ in that every removed degree-two path becomes a required representative with an explicit reconstruction map.

\paragraph{Representation and algorithm selection.}
Per-instance selection exploits complementary algorithms using cheaply computed instance features~\cite{kerschke2019algorithmselection}. Most closely, Liu and Truszczynski learn which of six ASP encodings to use for a directed Hamiltonian-cycle instance~\cite{liu2019encodingselection}. \method{} uses deterministic graph recognizers to select both the graph abstraction and the SAT variables. Complete branches handle inputs for which the compact candidate searches find no cycle.

\paragraph{Other exact methods.}
Dynamic programming provides another representation of the global connectivity constraint. Ziobro and Pilipczuk~\cite{ziobro2018treewidth} maintain partial paths along a tree decomposition, while Fieger~\cite{fieger2018parallel} combines exact longest-path states across graph partitions. Their results extend the comparison beyond SAT-based methods.

\paragraph{Empirical comparison.}
The closest SAT solvers differ mainly in how they enforce connectivity while encoding the input graph directly. Static encodings impose connectivity before search, SAT--CEGAR learns it from disconnected cycle covers, and HCup can intervene during partial assignments. \method{} instead changes the representation before the main search. We compare solved-set coverage with reruns of Cut-set CEGAR and HCup and use matched internal ablations to measure the effect of formulation selection~\cite{ohashi2025cutset,funder2026hcup,ziobro2018treewidth}.

\section{Method}
\label{sec:method}

\subsection{Overview}
\label{sec:overview}

\method{} first checks inexpensive necessary conditions and then chooses a SAT formulation from the graph structure. Specialized solvers handle graphs with high-degree clusters or a paired-chain structure. Cubic graphs use a directed encoding. The remaining graphs use an undirected edge encoding after maximal degree-two paths have been contracted. If a specialized solver finds no cycle, the general solver takes over.

Degree-two path contraction produces a reduced graph $H$, a set $F$ of required virtual edges, and a reconstruction map $\rho$. For a contracted path $P$, the virtual edge $\epsilon_P\in F$ is fixed to true in the SAT formula and $\rho(\epsilon_P)=P$ stores the ordered path that it represents. Replacing each virtual edge by its stored path maps a cycle in $H$ back to a cycle in the input graph.

Figure~\ref{fig:pipeline} shows the four formulations: a cluster formula, a directed pair graph, a directed-arc formula for cubic graphs, and a contracted undirected formula for the remaining inputs. Every successful route ultimately returns a cycle verified in the original graph; the complete cubic and general routes can also prove UNSAT.

\begin{figure*}[t]
  \centering
  \begin{tikzpicture}[
      x=1cm, y=1cm,
      >=Latex,
      flow/.style={draw=black!75, line width=0.55pt, rounded corners=3pt, fill=black!4, align=center,
        minimum height=0.68cm, font=\scriptsize\linespread{0.92}\selectfont},
      flowmain/.style={draw=black!90, line width=0.65pt, rounded corners=3pt, fill=black!7, align=center,
        minimum height=0.68cm, font=\scriptsize\bfseries\linespread{0.92}\selectfont},
      card/.style={draw=black!65, line width=0.55pt, rounded corners=4pt, fill=black!2,
        minimum width=3.80cm, minimum height=2.38cm},
      title/.style={font=\scriptsize\bfseries, align=center},
      badge/.style={font=\fontsize{6.5}{7.1}\selectfont\itshape, text=black!75, align=center},
      note/.style={font=\fontsize{6.5}{7.1}\selectfont, align=center, text width=3.55cm},
      vertex/.style={circle, fill=black!85, inner sep=1.2pt},
      hubvertex/.style={circle, fill=black, inner sep=1.8pt},
      edge/.style={line width=0.5pt, draw=black!75},
      path/.style={->, line width=0.65pt, draw=black!85},
      optbox/.style={draw=black!60, dashed, rounded corners=2pt, fill=white,
        font=\fontsize{6.5}{7.1}\selectfont, align=center, inner sep=1.8pt}
    ]

    \node[card] (c1) at (2.15, 1.25) {};
    \node[title] at (2.15, 2.22) {High-Degree Cluster};
    \node[badge] at (2.15, 2.00) {Structure-Specific Search};
    \node[hubvertex] (h_hub) at (1.25, 1.48) {};
    \node[vertex] (h_c1) at (2.10, 1.76) {};
    \node[vertex] (h_c2) at (2.10, 1.20) {};
    \node[vertex] (h_c3) at (2.85, 1.76) {};
    \node[vertex] (h_c4) at (2.85, 1.20) {};
    \draw[edge] (h_hub)--(h_c1) (h_hub)--(h_c2) (h_hub)--(h_c3) (h_hub)--(h_c4);
    \draw[edge] (h_c1)--(h_c2)--(h_c4)--(h_c3)--(h_c1) (h_c1)--(h_c4) (h_c2)--(h_c3);
    \node[note] at (2.15, 0.54) {Choose connections between clusters;\\find one spanning path per cluster};

    \node[card] (c2) at (6.45, 1.25) {};
    \node[title] at (6.45, 2.22) {Paired-Chain Structure};
    \node[badge] at (6.45, 2.00) {Structure-Specific Search};
    \node[vertex] (p_u1) at (5.65, 1.58) {};
    \node[vertex] (p_u2) at (6.45, 1.58) {};
    \node[vertex] (p_u3) at (7.25, 1.58) {};
    \node[vertex] (p_v1) at (5.65, 1.20) {};
    \node[vertex] (p_v2) at (6.45, 1.20) {};
    \node[vertex] (p_v3) at (7.25, 1.20) {};
    \draw[edge] (p_u1)--(p_v1) (p_u2)--(p_v2) (p_u3)--(p_v3);
    \draw[edge] (p_u1)--(p_u2)--(p_u3) (p_v1)--(p_v2)--(p_v3);
    \node[note] at (6.45, 0.54) {Directed pair graph with\\lazy directed cycle cuts};

    \node[card] (c3) at (10.75, 1.25) {};
    \node[title] at (10.75, 2.22) {Cubic Graph};
    \node[badge] at (10.75, 2.00) {Complete Directed Formulation};
    \node[vertex] (k_1) at (10.15, 1.60) {};
    \node[vertex] (k_2) at (11.35, 1.60) {};
    \node[vertex] (k_3) at (10.15, 1.18) {};
    \node[vertex] (k_4) at (11.35, 1.18) {};
    \draw[edge] (k_1)--(k_2)--(k_4)--(k_3)--(k_1) (k_1)--(k_4) (k_2)--(k_3);
    \node[note] at (10.75, 0.54) {Directed-arc SAT--CEGAR\\(exact 1-in/1-out constraints)};

    \node[card] (c4) at (15.05, 1.25) {};
    \node[title] at (15.05, 2.22) {General Graph};
    \node[badge] at (15.05, 2.00) {Complete Undirected Formulation};
    \node[vertex] (g_1) at (14.05, 1.62) {};
    \node[vertex] (g_2) at (14.70, 1.62) {};
    \node[vertex] (g_3) at (15.35, 1.62) {};
    \node[vertex] (g_4) at (16.00, 1.62) {};
    \draw[edge] (g_1)--(g_2)--(g_3)--(g_4);
    \draw[edge] (g_1)--(13.80, 1.82) (g_1)--(13.80, 1.42)
               (g_4)--(16.25, 1.82) (g_4)--(16.25, 1.42);
    \node[note] at (15.05, 1.14) {Degree-two path contraction\\(virtual edges $F$; reconstruction $\rho$)};
    \node[optbox, text width=3.45cm] (opt_twocut) at (15.05, 0.52)
      {Optional separation-pair reduction\\(ablation only)};

    \node[flowmain, minimum width=16.7cm, minimum height=0.64cm, anchor=south] (select) at ($(c1.north)+(6.45, 0.22)$) {};
    \node[font=\scriptsize\bfseries] at ($(select.center)+(0, 0.15)$) {Select the First Applicable Route};
    \node[font=\fontsize{6.5}{7.1}\selectfont\bfseries, text=black!90] at ($(select.center)+(-6.45, -0.13)$) {(1) High-Degree Clusters};
    \node[font=\fontsize{6.5}{7.1}\selectfont, text=black!55] at ($(select.center)+(-4.30, -0.13)$) {$\longrightarrow$};
    \node[font=\fontsize{6.5}{7.1}\selectfont\bfseries, text=black!90] at ($(select.center)+(-2.15, -0.13)$) {(2) Paired Chains};
    \node[font=\fontsize{6.5}{7.1}\selectfont, text=black!55] at ($(select.center)+(0, -0.13)$) {$\longrightarrow$};
    \node[font=\fontsize{6.5}{7.1}\selectfont\bfseries, text=black!90] at ($(select.center)+(2.15, -0.13)$) {(3) Cubic Graph};
    \node[font=\fontsize{6.5}{7.1}\selectfont, text=black!55] at ($(select.center)+(4.30, -0.13)$) {$\longrightarrow$};
    \node[font=\fontsize{6.5}{7.1}\selectfont\bfseries, text=black!90] at ($(select.center)+(6.45, -0.13)$) {(4) General Graph (Default)};

    \node[flow, minimum width=10.20cm, text width=9.95cm, anchor=south] (filter) at ($(select.north)+(1.10, 0.38)$)
      {\textbf{Necessary-Condition Checks}\\
       {\fontsize{6.5}{7.1}\selectfont connected; $\delta(G)\geq2$; biconnected; balanced if bipartite $\;\longrightarrow\;$ else \textsc{unsat}}};
    \node[flowmain, minimum width=2.80cm, anchor=south] (input) at ($(select.north)+(-6.00, 0.38)$) {Input Graph $G$};

    \draw[path] (input) -- (filter);
    \draw[path] (8.60, 0 |- filter.south) -- (8.60, 0 |- select.north)
      node[pos=0.50, right, font=\fontsize{6.5}{7.1}\selectfont\bfseries] {passes checks};

    \draw[path] (c1.north |- select.south) -- (c1.north);
    \draw[path] (c2.north |- select.south) -- (c2.north);
    \draw[path] (c3.north |- select.south) -- (c3.north);
    \draw[path] (c4.north |- select.south) -- (c4.north);

    \node[flow, minimum width=7.30cm, minimum height=0.78cm, anchor=north] (special_stage) at ($(c1.south)+(2.15, -0.22)$)
      {\textbf{Structure-Specific Search \& Path Expansion}\\{\fontsize{6.5}{7.1}\selectfont Returns a candidate; otherwise uses the general branch}};

    \node[flow, minimum width=7.30cm, minimum height=0.78cm, anchor=north] (complete_stage) at ($(c1.south)+(10.75, -0.22)$)
      {\textbf{Complete Incremental SAT--CEGAR Solving}\\{\fontsize{6.5}{7.1}\selectfont \textbf{Cubic:} directed CEGAR $\;\mid\;$ \textbf{General:} undirected CEGAR + restore $\rho$}};

    \draw[dashed, ->, line width=0.55pt, draw=black!75] (special_stage.east) --
      node[above, font=\fontsize{6.5}{7.1}\selectfont\bfseries\sffamily, yshift=0.5pt] {\textsc{fallback}} (complete_stage.west);

    \draw[path] (c1.south) -- (c1.south |- special_stage.north);
    \draw[path] (c2.south) -- (c2.south |- special_stage.north);
    \draw[path] (c3.south) -- (c3.south |- complete_stage.north);
    \draw[path] (c4.south) -- (c4.south |- complete_stage.north);

    \node[flowmain, fill=black!8, minimum width=7.30cm, minimum height=0.74cm, anchor=north] (verified) at ($(special_stage.south)+(0, -0.72)$)
      {\textbf{Candidate Verification on Input Graph $G$}\\{\scriptsize Every vertex appears once and every cycle edge belongs to $G$ $\;\longrightarrow\;$ \textbf{Hamiltonian Cycle}}};

    \node[flowmain, fill=black!8, minimum width=7.30cm, minimum height=0.74cm, anchor=north] (decision) at ($(complete_stage.south)+(0, -0.72)$)
      {\textbf{Complete-Branch Exact Decision}\\{\scriptsize \textsc{unsat} $\;\mid\;$ \textsc{unknown} at the resource limit}};

    \coordinate (bus_pt) at ($(special_stage.south)!0.5!(verified.north)$);
    \draw[line width=0.65pt, draw=black!85] (special_stage.south) -- (bus_pt)
      node[pos=0.50, left, font=\fontsize{6.5}{7.1}\selectfont\bfseries] {expanded cycle};

    \draw[line width=0.65pt, draw=black!85] ($(complete_stage.south)+(-1.70,0)$) -- ++(0,-0.36) -- (bus_pt);
    \node[font=\fontsize{6.5}{7.1}\selectfont\bfseries, fill=white, inner sep=2pt] at ($(bus_pt)+(3.45, 0)$) {candidate cycle (Cubic / restored $\rho$)};

    \draw[path] (bus_pt) -- (verified.north);
    \draw[path] (complete_stage.south) -- (decision.north);

  \end{tikzpicture}
  \caption{Overview of \method{}. After necessary-condition checks, the solver follows the first applicable structural route. The high-degree-cluster and paired-chain routes construct and expand candidate cycles; if neither succeeds, the complete general branch continues the search. Cubic graphs use directed SAT--CEGAR without reconstruction. Other graphs undergo degree-two path contraction, undirected SAT--CEGAR, and path reconstruction via $\rho$. Every candidate is verified on the input graph. Separation-pair reduction appears as an optional ablation. Complete branches decide \textsc{unsat} or report \textsc{unknown} at the resource limit.}
  \label{fig:pipeline}
  \Description{A four-tier pipeline maps an input graph through necessary-condition checks and an ordered structure recognizer to high-degree-cluster, paired-chain, cubic, or general solving routes. The first two routes can produce a candidate cycle but otherwise fall back to the complete general branch. The cubic route uses directed SAT-CEGAR without reconstruction. The general route contracts degree-two paths, uses undirected SAT-CEGAR, and restores paths. Every candidate cycle is verified in the input graph. Complete branches may return UNSAT or UNKNOWN under a resource limit. Separation-pair reduction appears only as a disabled, ablation-only option.}
\end{figure*}
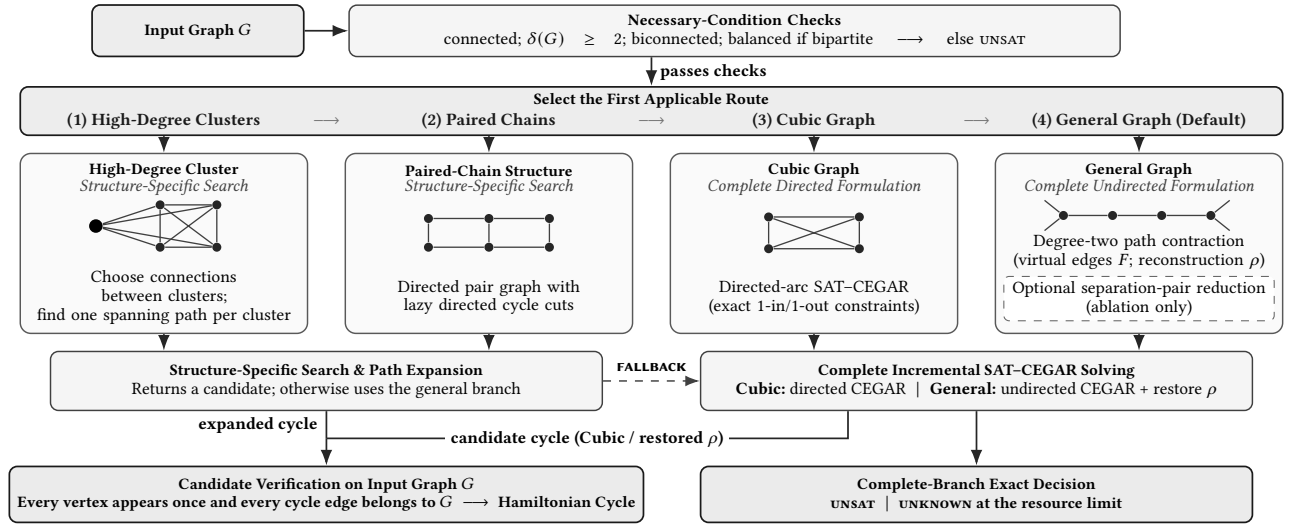

\subsection{Graph Reductions}
\label{sec:reductions}

\paragraph{Degree-two path contraction.}
Let $P=(u,v_1,\ldots,v_r,w)$ be a maximal path whose internal vertices have degree two. Every Hamiltonian cycle contains both incident edges of each $v_i$ and therefore traverses $P$ from $u$ to $w$. We replace the internal vertices by a representative edge $\epsilon_P=(u,w)$, require $x_{\epsilon_P}=1$, and store the ordered internal path in a map $\rho$. Connected degree-two components that form an entire cycle are reduced to a triangle rather than below three vertices; self-loop-producing contractions are rejected.

Contraction may create several edges with the same endpoints. A direct shortcut is dominated by the forced path because selecting both would close a proper cycle. When distinct forced paths share an endpoint pair, reconstruction checks the expanded cycle before acceptance. An incomplete expansion is discarded and the search continues on the original graph.

\begin{proposition}[Degree-two path contraction]
\label{prop:series}
Suppose each contracted path has a unique endpoint key. Let $H$ be the resulting graph and $F$ its required representative edges. Then $G$ has a Hamiltonian cycle if and only if $H$ has one containing every edge in $F$.
\end{proposition}
\begin{proof}
Contracting each forced traversal in a Hamiltonian cycle of $G$ yields a cycle of $H$ containing $F$; a coincident original shortcut is never selected by the argument above. Conversely, unique endpoint keys make the paths stored by $\rho$ internally disjoint. Replacing each representative with its oriented path therefore visits every removed vertex once and uses only edges of $G$.
\end{proof}

\paragraph{Optional separation-pair reduction.}
The optional variant identifies separation pairs $\{u,v\}$ by computing the articulation vertices of $G-\{u\}$ after degree-two contraction. If deleting $u$ and $v$ leaves more than two components, the graph is non-Hamiltonian because deleting two vertices from a cycle leaves at most two paths. When exactly two components remain, the smaller component is tried first as a spanning $u$--$v$ path subproblem.

\begin{proposition}[Separation-pair equivalence]
\label{prop:twocut}
If $G-\{u,v\}$ has exactly two components $C_1,C_2$, then $G$ is Hamiltonian if and only if each $G[C_i\cup\{u,v\}]$ has a spanning $u$--$v$ path. A solved path through one component can therefore be replaced by a required virtual edge $(u,v)$.
\end{proposition}
\begin{proof}
Removing $u$ and $v$ from a Hamiltonian cycle leaves two vertex-disjoint paths. Because no edge joins $C_1$ and $C_2$, each path lies in one component and has endpoints $u,v$. Conversely, the union of two internally disjoint spanning $u$--$v$ paths is a Hamiltonian cycle.
\end{proof}

A successful path solve removes the component $C$ and stores the path as $\rho((u,v))$. The reduced formula replaces the edge variables incident to $C$ with one required-edge variable. If no spanning path is found within the local search budget, the graph is left unchanged. Repeated successful reductions preserve Hamiltonicity by Propositions~\ref{prop:series} and~\ref{prop:twocut}, and stored paths are restored in reverse order. This option is disabled in the main configuration because its ablation reduces coverage.

\subsection{SAT Encodings and Formulation Selection}
\label{sec:engine}

The four branches in Figure~\ref{fig:pipeline} solve the same problem, but expose different choices to SAT. The first three branches use a compact representation for a recognizable graph family. The general branch keeps one variable for each undirected edge.

\paragraph{High-degree clusters.}
Let $Q=\{v:\deg(v)\geq400\}$ be the set of hubs. The recognizer applies when $|Q|$ is 5 or 10 and builds a deterministic partition in three steps. First, it identifies a low-degree connector corridor $R$ containing $Q$. A non-hub of degree at most four seeds this corridor when it has at least two hub neighbors, no hub neighbor, or at least two low-degree neighbors; its low-degree neighbors are included as connectors. Second, each remaining vertex is assigned to its unique adjacent hub when possible, or to the hub owning most of its neighbors. These assignments form bulks $K_h$. Finally, a vertex in $K_h$ becomes a port in $B_h$ when it connects to the corridor or another bulk.

The macro graph contains $R\cup\bigcup_h B_h$ and original edges whose endpoints do not share a bulk owner. It has an edge variable for every macro edge and a pair variable $p^h_{ab}$ for each $a,b\in B_h$: selecting $p^h_{ab}$ promises a spanning $a$--$b$ path through $K_h$. Exactly one pair is selected per bulk. Port clauses match a selected pair endpoint to one external macro edge; corridor vertices have macro degree two. CEGAR cut clauses join disconnected macro cycles. For each selected pair, an independent edge-level SAT call seeks the promised spanning path; a failed pair is blocked and the macro search continues. Only after all paths expand to a cycle accepted by the input-graph verifier does this route succeed. The route constructs candidates within fixed cut and iteration limits, while complete decisions remain with the cubic and general branches.

\paragraph{Paired chains.}
This recognizer first requires at least $15\%$ degree-two vertices and contracts their paths. Their endpoint pairs must partition all vertices of the contracted graph. In the auxiliary graph containing both pair edges and all remaining edges, a breadth-first two-coloring must place pair partners on opposite sides and every non-pair edge across the sides. Pair $i$ is then oriented from its color-0 endpoint $s_i$ to its color-1 endpoint $t_i$. A directed arc $a_{ij}$ exists exactly when the contracted graph contains the non-pair edge $(t_i,s_j)$.

The formula selects exactly one incoming and one outgoing arc at every pair, forbids selecting both $a_{ij}$ and $a_{ji}$, and initially blocks directed 3-cycles. Proper cycles found in a model generate internal-cycle nogoods and incoming/outgoing cut clauses. Parallel SAT workers share the accumulated cuts and deadline. A single directed cycle expands each oriented pair and its stored path before input-graph verification. If the search reaches its round or time limit, the complete general branch continues.

\paragraph{Cubic graphs.}
When every vertex has degree three, let $N(u)$ denote the three neighbors of $u$. Each undirected edge $\{u,v\}$ is represented by two directed variables: $y_{uv}=1$ means that the candidate cycle traverses the edge from $u$ to $v$, and $y_{vu}$ represents the reverse orientation. The constraints
\begin{equation}
 \begin{aligned}
  \sum_{v\in N(u)}y_{uv}&=1 && \text{for every }u\in V,\\
  \sum_{v\in N(u)}y_{vu}&=1 && \text{for every }u\in V,\\
  y_{uv}+y_{vu}&\leq 1 && \text{for every }\{u,v\}\in E.
 \end{aligned}
 \label{eq:cubic}
\end{equation}
assign every vertex exactly one successor and one predecessor, while the final inequality forbids traversing one undirected edge in both directions. A model may still contain several directed cycles. For every proper cycle on vertex set $S$, the solver therefore adds $\bigvee_{u\in S,\,v\notin S}y_{uv}$, requiring the next model to leave $S$. The process ends with one directed Hamiltonian cycle or an UNSAT proof.

\paragraph{General graphs.}
The general encoding introduces one variable per undirected edge and enforces Equation~\eqref{eq:degree}. Direct clauses encode at-most-two constraints for low-degree vertices; larger neighborhoods use a sequential counter. Unit clauses select the virtual edges in $F$. In the optional separation-pair variant, the same engine solves a spanning $s$--$t$ path by adding a fresh degree-two vertex adjacent only to $s$ and $t$ and deleting that vertex from the resulting cycle.

The general CEGAR loop repeatedly calls SAT, decodes the selected edges, and returns when the model is one spanning cycle. Suppose a disconnected model contains a proper cycle $C$, with vertex set $V(C)$ and edge set $E(C)$. The engine adds two valid clauses,
\begin{equation}
 \bigvee_{e\in\delta_H(V(C))}x_e,
 \qquad
 \bigvee_{e\in E(C)}\neg x_e.
 \label{eq:refinement}
\end{equation}
Here, $\delta_H(V(C))$ is the cut of the current reduced graph $H$. The first clause requires a later model to select an edge leaving $V(C)$; the second requires it to break the current internal cycle. CaDiCaL remains incremental across refinements. Before the next SAT call, 2-opt and restricted 3-opt exchanges attempt to merge the cycles in the current model while preserving required virtual edges.

\paragraph{Selection order.}
The solver chooses the first applicable route in Figure~\ref{fig:pipeline}. If the high-degree or paired-chain route finds no cycle, execution proceeds to the general branch. A graph reaching the cubic test uses Equation~\eqref{eq:cubic}; every other graph uses degree-two contraction and the complete undirected encoding. All routes share one deadline, after which the solver returns \textsc{unknown}.

\subsection{Correctness and Complexity}
\label{sec:correctness}

Algorithm~\ref{alg:unified} summarizes \method. The return states distinguish an inapplicable or failed candidate route from a proof of UNSAT. A resource limit produces UNKNOWN.

\begin{algorithm}[t]
  \caption{\method}
  \label{alg:unified}
  \small
  \begin{algorithmic}[1]
    \Require Graph $G$ and deadline $D$
    \Ensure Hamiltonian cycle, \textsc{unsat}, or \textsc{unknown}
    \If{a necessary-condition filter fails}
      \State \Return \textsc{unsat}
    \EndIf
    \State $r\gets\Call{FirstApplicableRoute}{G}$
    \If{$r\in\{\textsc{high-degree},\textsc{paired-chain}\}$}
      \State $z\gets\Call{TryCandidateRoute}{r,G,D}$
      \If{$z\neq\bot$ \textbf{and} $\Call{VerifyCycle}{G,z}$} \State \Return $z$ \EndIf
      \If{$\Call{Now}{}\geq D$} \State \Return \textsc{unknown} \EndIf
    \ElsIf{$r=\textsc{cubic}$}
      \State $z\gets\Call{DirectedCycleCegar}{G,D}$
      \If{$z$ is a cycle \textbf{and} $\neg\Call{VerifyCycle}{G,z}$} \State \Return \textsc{unknown} \EndIf
      \State \Return $z$
    \EndIf
    \State $(H,F,\rho)\gets\Call{ContractDegreeTwoPaths}{G}$
    \State $z\gets\Call{UndirectedCycleCegar}{H,F,D}$
    \If{$z\in\{\textsc{unsat},\textsc{unknown}\}$} \State \Return $z$ \EndIf
    \State $z'\gets\Call{Reconstruct}{z,\rho}$
    \If{$|z'|=|V|$ \textbf{and} $\Call{VerifyCycle}{G,z'}$} \State \Return $z'$ \EndIf
    \State \Return $\Call{OriginalGraphCegarAndVerify}{G,D}$
  \end{algorithmic}
\end{algorithm}

The optional separation-pair variant applies the reduction in Section~\ref{sec:reductions} after degree-two path contraction. Proposition~\ref{prop:twocut} gives the corresponding equivalence. The clauses in Equation~\eqref{eq:refinement} preserve every Hamiltonian cycle.

\begin{proposition}[Specialized branches preserve exactness]
\label{prop:specialized}
A cycle returned by either specialized branch is Hamiltonian in the input graph. If a specialized search finds no cycle, the complete general encoding receives the same instance while time remains.
\end{proposition}
\begin{proof}
Each specialized branch expands its compact solution into an ordered list of original vertices. The verifier checks that the list has length $|V|$, contains every vertex once, and uses an input edge between every consecutive pair, including the closing pair. Algorithm~\ref{alg:unified} sends all other outcomes to complete general CEGAR or returns UNKNOWN at the deadline.
\end{proof}

\begin{theorem}[End-to-end correctness]
\label{thm:pipeline}
Without a resource limit, \method{} returns a Hamiltonian cycle exactly when the input graph is Hamiltonian; otherwise it returns UNSAT. Every returned cycle is valid in the original graph.
\end{theorem}
\begin{proof}
The initial filters test necessary conditions, and Proposition~\ref{prop:specialized} covers the candidate routes. The directed cubic formula is a cycle-cover encoding refined only by clauses valid for every directed Hamiltonian cycle. In the general branch, Proposition~\ref{prop:series} establishes equivalence for unique endpoint keys. Endpoint collisions either yield a verified reconstruction or redirect the search to the original graph. Each CEGAR refinement preserves every Hamiltonian cycle. Because each complete formula has finitely many assignments, an unrestricted run terminates with a verified cycle or an UNSAT proof.
\end{proof}

For a graph with $m$ edges, the undirected engine starts with $m$ primary variables. Direct at-most-two clauses require $O(d^3)$ clauses at a low-degree vertex of degree $d$; the sequential encoding used for larger degrees is linear in $d$. If the high-degree partition has macro-edge set $E_M$ and port sets $B_h$, its macro formula has $|E_M|+\sum_h\binom{|B_h|}{2}$ primary variables, in addition to the local path calls. With $q$ paired chains and arc set $A$, the paired formula has $|A|$ variables; the cubic formula has $2m$ directed variables. Degree-two contraction and reconstruction are linear in $n+m$. The optional separation-pair scan is $O(n(n+m))$ in the worst case. SAT search and the number of refinements remain exponential in the worst case.

\section{Evaluation}
\label{sec:experiments}

We measure standalone and complementary coverage against recent SAT--CEGAR solvers, isolate the contribution of each structural choice, and test the cubic formulation on both Hamiltonian and non-Hamiltonian graphs.

\subsection{Benchmarks and Compared Configurations}
\label{sec:setup}

\method{} is implemented with RustSAT and CaDiCaL~\cite{jabs2025rustsat,biere2024cadical}.

The FHCP Challenge Set contains 1,001 known-Hamiltonian graphs~\cite{haythorpe2018fhcp}. The matched ablations ran on Google Cloud c4-highcpu-8 Linux instances with 8 logical CPUs, 16~GB RAM, and Ubuntu 22.04. We ran one solver process at a time with at most three internal workers. Each configuration used seed 1, a 1,800-second wall-time cutoff, and a 7,168-MB memory limit enforced by Runlim 2.0.0rc14. The two recognizer thresholds were fixed before the final runs. A run counted as solved only after independent tour verification.

\paragraph{Metrics.}
We summarize verified coverage, penalized runtime, solved-run resource use, and complementary coverage. Let $\mathcal I$ be the $N=1{,}001$ instances, $T=1{,}800$ seconds, and $S_s$ the instances solved by method $s$. Its penalized cost is
\[
  c_{s,i}=\begin{cases}
    t_{s,i}, & i\in S_s,\\
    2T, & i\notin S_s,
  \end{cases}
  \qquad
  \operatorname{PAR2}(s)=\frac{1}{N}\sum_{i\in\mathcal I}c_{s,i}.
\]
Thus every timeout, memory-out, or unverified outcome contributes $2T=3{,}600$ seconds. PAR-2 combines coverage and runtime; solved-run medians and interquartile ranges describe runtime conditional on success.

For comparison set $C$, the union is $U_C=\bigcup_{r\in C}S_r$ and method $s$ contributes $|S_s\setminus U_{C\setminus\{s\}}|$ exclusive solves. Relative to a prior union $U_P$, its additional coverage is $|S_s\setminus U_P|$. Cactus endpoints equal solved counts. All aggregates use unrounded per-instance records.

We compare solved-set coverage with archived reruns of two recent methods. HCup's outside-mergecutset configuration merges adjacent subcycles and adds cut-set clauses~\cite{funder2026hcup}. Ohashi et al.'s Cut-set CEGAR extends the earlier HCP CEGAR loop with cut-set refinement. Its \texttt{all} configuration also enables two-cycle prevention, orientation symmetry breaking, and subcycle merging~\cite{ohashi2025cutset}. The selected runs each solve 937 instances and form a 950-instance union. Because the archives use different runtime protocols, we use their verified solved sets to measure complementary coverage. Available timing profiles are reported as contextual information.

The published SAT-based results of Ohashi et al.~\cite{ohashi2025cutset} and Funder~\cite{funder2026hcup} solve 972 distinct instances. Treewidth-based dynamic programming~\cite{ziobro2018treewidth} adds six, bringing the union of published exact results to 978 instances.

We also use the labelled cubic-graph census from the Flinders graph database~\cite{fhcpGraphDatabase}. It contains all 1,927 non-Hamiltonian graphs of orders 10--18. From the Hamiltonian graphs we draw, independently at each order, a sample matching the number of biconnected negatives, using sampling seed 2027; the solver uses seed 1.

\paragraph{Reproducibility artifact.}
The anonymized supplementary artifact provides the filtered per-instance records underlying the reported configurations, together with a script that checks completeness and recomputes the principal aggregate outcomes.

\subsection{Comparison with Prior Solvers}
\label{sec:results}

Table~\ref{tab:full-coverage}(a) compares the three methods by coverage, failure mode, and exclusive contribution to their union. Panel (b) reports the runtime profile for the two methods with complete per-instance timing data.

\begin{table}[H]
  \caption{FHCP comparison. Panel (a) reports run outcomes and exclusive coverage in the three-method union; TO, MO, and E denote timeout, memory-out, and other error. Panel (b) reports the two archives with per-instance timing.}
  \label{tab:full-coverage}
  \centering\footnotesize
  \setlength{\tabcolsep}{3pt}
  \textbf{(a) Outcomes and complementary coverage}\par\smallskip
  \begin{tabular*}{\columnwidth}{@{\extracolsep{\fill}}p{.48\columnwidth}rrr@{}}\toprule
    Method & Solved & Failures & Excl.\\ \midrule
    \method{} (ours) & 826 & 30 TO, 145 MO & \textbf{34}\\
    HCup (\texttt{outside-mergecutset}) & \textbf{937} & 63 TO, 1 E & 5\\
    Cut-set CEGAR (\texttt{all}) & \textbf{937} & 64 TO & 1\\
    \bottomrule
  \end{tabular*}
  \par\medskip
  \textbf{(b) Runtime profile}\par\smallskip
  \begin{tabular*}{\columnwidth}{@{\extracolsep{\fill}}p{.30\columnwidth}p{.31\columnwidth}rrr@{}}\toprule
    Method & $t_{50}$ [IQR] & $\leq10$ s & $\leq100$ s & PAR-2\\ \midrule
    \method{} & 2.69 [0.37, 10.25] & \textbf{617} & 785 & 649.4\\
    Cut-set CEGAR (\texttt{all}) & 4.24 [0.37, 38.12] & 563 & \textbf{798} & \textbf{288.4}\\
    \bottomrule
  \end{tabular*}
\end{table}

HCup and Cut-set CEGAR each solve 937 instances. Their union covers 950 and leaves 51 common misses. \method{} solves 34 of these, raising the three-method union to 984 (98.3\%); the two comparators contribute six exclusive solves in the same union.

In the archived timing traces, \method{} has 617 verified cycles by 10 seconds versus 563. Cut-set CEGAR leads by 100 seconds, solves 111 more instances overall, and attains the lower PAR-2. The solved sets remain complementary: \method{} contributes 34 independently verified cycles outside the two-solver union.

\paragraph{Where the additional coverage occurs.}
Figure~\ref{fig:outcome-profile}(a) groups the full benchmark by average degree, showing that \method{} solves all 160 intermediate-density graphs and 84 of 92 dense graphs. Among the 51 common misses in Figure~\ref{fig:outcome-profile}(b), it solves all 22 dense graphs and 12 of 29 sparse graphs; all 17 three-way misses are sparse. The additional coverage is therefore concentrated in the dense cohort recognized by the high-degree route.

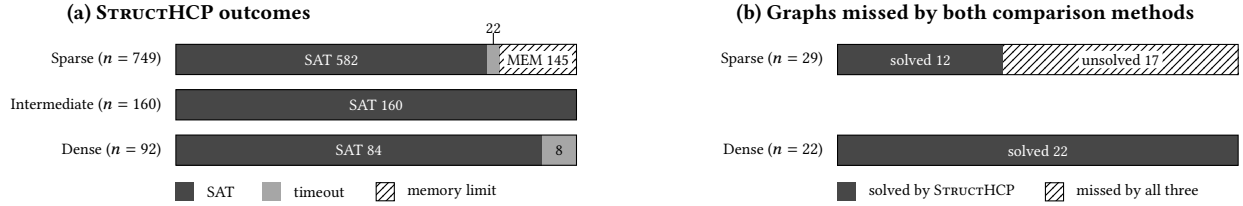
\begin{figure*}[t]
  \centering
  \begin{minipage}{0.49\textwidth}
    \centering
    \begin{tikzpicture}[font=\scriptsize]
      \node[anchor=west,font=\small\bfseries] at (-1.55,1.85)
        {(a) \method{} outcomes};

      \node[anchor=east] at (-0.08,1.25) {Sparse ($n=749$)};
      \fill[black!72] (0,1.05) rectangle (4.12,1.45);
      \fill[black!35] (4.12,1.05) rectangle (4.28,1.45);
      \fill[pattern=north east lines] (4.28,1.05) rectangle (5.30,1.45);
      \draw (0,1.05) rectangle (5.30,1.45);
      \node[text=white] at (2.06,1.25) {SAT 582};
      \node[fill=white,inner sep=1.2pt,rounded corners=1pt] at (4.79,1.25) {MEM 145};
      \draw[-] (4.20,1.45) -- (4.20,1.60);
      \node[anchor=south,inner sep=0.5pt] at (4.20,1.58) {22};

      \node[anchor=east] at (-0.08,0.65) {Intermediate ($n=160$)};
      \fill[black!72] (0,0.45) rectangle (5.30,0.85);
      \draw (0,0.45) rectangle (5.30,0.85);
      \node[text=white] at (2.65,0.65) {SAT 160};

      \node[anchor=east] at (-0.08,0.05) {Dense ($n=92$)};
      \fill[black!72] (0,-0.15) rectangle (4.84,0.25);
      \fill[black!35] (4.84,-0.15) rectangle (5.30,0.25);
      \draw (0,-0.15) rectangle (5.30,0.25);
      \node[text=white] at (2.42,0.05) {SAT 84};
      \node at (5.07,0.05) {8};

      \fill[black!72] (0,-0.62) rectangle (0.24,-0.38);
      \node[anchor=west] at (0.30,-0.50) {SAT};
      \fill[black!35] (1.15,-0.62) rectangle (1.39,-0.38);
      \node[anchor=west] at (1.45,-0.50) {timeout};
      \fill[pattern=north east lines] (2.65,-0.62) rectangle (2.89,-0.38);
      \draw (2.65,-0.62) rectangle (2.89,-0.38);
      \node[anchor=west] at (2.95,-0.50) {memory limit};
    \end{tikzpicture}
  \end{minipage}\hfill
  \begin{minipage}{0.49\textwidth}
    \centering
    \begin{tikzpicture}[font=\scriptsize]
      \node[anchor=west,font=\small\bfseries] at (-1.45,1.85)
        {(b) Graphs missed by both comparison methods};

      \node[anchor=east] at (-0.08,1.25) {Sparse ($n=29$)};
      \fill[black!72] (0,1.05) rectangle (2.19,1.45);
      \fill[pattern=north east lines] (2.19,1.05) rectangle (5.30,1.45);
      \draw (0,1.05) rectangle (5.30,1.45);
      \node[text=white] at (1.095,1.25) {solved 12};
      \node[fill=white,inner sep=1.2pt,rounded corners=1pt] at (3.745,1.25) {unsolved 17};

      \node[anchor=east] at (-0.08,0.05) {Dense ($n=22$)};
      \fill[black!72] (0,-0.15) rectangle (5.30,0.25);
      \draw (0,-0.15) rectangle (5.30,0.25);
      \node[text=white] at (2.65,0.05) {solved 22};

      \fill[black!72] (0,-0.62) rectangle (0.24,-0.38);
      \node[anchor=west] at (0.30,-0.50) {solved by \method};
      \fill[pattern=north east lines] (2.75,-0.62) rectangle (2.99,-0.38);
      \draw (2.75,-0.62) rectangle (2.99,-0.38);
      \node[anchor=west] at (3.05,-0.50) {missed by all three};
    \end{tikzpicture}
  \end{minipage}
  \caption{Descriptive outcome profile by average degree $\bar d=2|E|/|V|$: sparse ($\bar d<4$), intermediate ($4\leq\bar d<6$), and dense ($\bar d\geq6$). Bars are normalized within each row; labels give instance counts.}
  \label{fig:outcome-profile}
  \Description{Two normalized bar charts. The first shows SAT, timeout, and memory-limit counts across three average-degree bands. The second shows solved and unsolved instances among graphs missed by both comparison methods.}
\end{figure*}

\paragraph{Broader published coverage.}
Table~\ref{tab:cohorts} compares \method{} with the per-instance unions available from prior work.

\begin{table}[t]
  \caption{Additional FHCP coverage supplied by \method{} beyond increasingly broad prior result sets. Published rows aggregate the cited per-instance result tables.}
  \label{tab:cohorts}
  \centering\small
  \begin{tabular*}{\columnwidth}{@{\extracolsep{\fill}}lrrr@{}}\toprule
    Comparison set & Prior union & Added & Final union\\ \midrule
    Two rerun configurations & 950 & \textbf{34} & \textbf{984}\\
    Published SAT results & 972 & \textbf{12} & \textbf{984}\\
    Published exact results & 978 & \textbf{6} & \textbf{984}\\
    \bottomrule
  \end{tabular*}
\end{table}

Twelve solutions lie outside the union of published per-instance SAT results included in our analysis. Treewidth dynamic programming covers six of these, while the other six also lie outside the published exact result sets in our comparison. Of the 29 graphs outside the published SAT union, \method{} solves all six dense graphs and six of 23 sparse graphs. The 17 unresolved graphs have average degree 3.16--3.42. Table~\ref{tab:residual-cohort} in Appendix~\ref{app:residual-cohort} reports these 29 outcomes instance by instance.

Having established the external coverage, we next isolate which structural choices produce it (Table~\ref{tab:configurations}).

\begin{table*}[t]
  \caption{FHCP ablation. $t_{50}$ [IQR] is solved-run time (s), $M_{50}/M_{95}$ is solved-run peak memory (MB), and TO/MO denote timeout/memory-out. Panel (b) gives net stepwise differences; panel (c) gives paired solved-set differences.}
  \label{tab:configurations}
  \centering\footnotesize
  \textbf{(a) Components and aggregate outcomes}\par\smallskip
  \setlength{\tabcolsep}{3.8pt}
  \begin{tabular*}{\textwidth}{@{\extracolsep{\fill}}p{.18\textwidth}cccrrrp{.15\textwidth}p{.08\textwidth}r@{}}\toprule
    & \multicolumn{3}{c}{Structural components} & \multicolumn{6}{c}{Outcomes and resources}\\
    \cmidrule(lr){2-4}\cmidrule(l){5-10}
    Configuration & \shortstack{Formulation\\selection} & \shortstack{Degree-2\\contraction} & \shortstack{Sep.-pair\\reduction} & Solved & TO & MO & $t_{50}$ [IQR] & $M_{50}/M_{95}$ & PAR-2\\ \midrule
    \method{} (main) & on & on & off & \textbf{826} & 30 & 145 & 2.69 [0.37, 10.25] & 36/215 & \textbf{649.4}\\
    No formulation selection & off & on & off & 744 & 108 & 149 & 4.08 [0.76, 13.18] & 39/409 & 956.1\\
    Unreduced-graph ablation & off & off & off & 741 & 105 & 155 & 3.32 [0.33, 12.62] & 40/359 & 967.4\\
    \method{} + separation-pair reduction & on & on & on & 818 & 29 & 154 & 2.45 [0.52, 11.13] & 21/211 & 680.8\\
    \bottomrule
  \end{tabular*}
  \par\medskip
  \textbf{(b) Stepwise localization by graph family}\par\smallskip
  \begin{tabular*}{\textwidth}{@{\extracolsep{\fill}}lrrrrrrr@{}}\toprule
    Family & $n$ & \shortstack{\method\\solved} & \shortstack{No selection\\solved} & \shortstack{Unreduced graph\\solved} & $\Delta_{\mathrm{sel}}$ & $\Delta_{\mathrm{ctr}}$ & \shortstack{Main/no selection\\PAR-2}\\ \midrule
    High-degree cluster & 69 & 65 & 6 & 7 & $+59$ & $-1$ & 247.8/3291.5\\
    Paired chain & 30 & 27 & 27 & 27 & $0$ & $0$ & 361.2/361.0\\
    Cubic & 88 & 88 & 66 & 67 & $+22$ & $-1$ & 0.1/1042.6\\
    General & 814 & 646 & 645 & 640 & $+1$ & $+5$ & 764.2/770.7\\
    \bottomrule
  \end{tabular*}
  \par\medskip
  \textbf{(c) Paired solved-set differences}\par\smallskip
  \begin{tabular*}{\textwidth}{@{\extracolsep{\fill}}lrrrr@{}}\toprule
    Contrast (A versus B) & Both solved & A only & B only & Net $A-B$\\ \midrule
    \method{} vs. no formulation selection & 736 & 90 & 8 & $+82$\\
    No formulation selection vs. unreduced graph & 732 & 12 & 9 & $+3$\\
    \method{} vs. + separation pair & 815 & 11 & 3 & $+8$\\
    \bottomrule
  \end{tabular*}
\end{table*}

\subsection{Ablation: What Drives \method?}
\label{sec:ablation}

The ablation removes one structural choice at a time. The main configuration selects a formulation by graph family and contracts maximal degree-two paths in its general branch. \emph{No formulation selection} sends every graph to the general undirected encoding but retains degree-two contraction. The \emph{unreduced-graph ablation} goes one step further: it applies the same general SAT--CEGAR engine directly to every input graph, without formulation selection or degree-two contraction. Finally, \emph{\method{} + separation-pair reduction} adds that optional reduction to the main method. Table~\ref{tab:configurations} therefore separates the effects of formulation selection, degree-two contraction, and separation-pair reduction.

\paragraph{Formulation selection is the main source of improvement.}
With degree-two contraction fixed in both configurations, formulation selection raises coverage from 744 to 826, a net 82 instances (8.2 percentage points), and lowers PAR-2 from 956.1 to 649.4. \method{} solves 90 instances missed without selection, while the no-selection configuration solves eight instances missed by \method{}. By 10 seconds, the two configurations solve 617 and 529 instances; by 100 seconds, they solve 785 and 694. On their 736 shared solves, \method{} is faster on 502 and slower on 162.

\paragraph{The gain is concentrated in the specialized formulations.}
Table~\ref{tab:configurations}(b) shows net differences of 59 solves in the high-degree-cluster family and 22 in the cubic family. These families account for 81 of the observed 82-solve difference. Paired-chain coverage is unchanged, while the general family differs by one solve. Among solved runs, the median maximum CEGAR iteration count falls from 130 without selection to 71.5 with \method{}. This reduction suggests that the specialized representations expose useful global choices earlier in the search.

Of the twelve cycles added to the union of published SAT results, the no-selection configuration retains five. The configuration with formulation selection solves the remaining seven, including all six dense instances handled by the high-degree-cluster route.

\paragraph{Degree-two contraction affects the failure tail.}
Comparing \emph{No formulation selection} with the \emph{unreduced-graph ablation} isolates the effect of degree-two contraction. Contraction changes coverage from 741 to 744: twelve instances become solved and nine are lost. The five-solve gain in the general family is offset by one loss in each of the high-degree and cubic cohorts. PAR-2 falls from 967.4 to 956.1, although contraction is slower on 492 of the 732 jointly solved instances and faster on 206. Its aggregate benefit therefore comes from a small number of cutoff outcomes rather than a uniform speedup.

\paragraph{Separation pairs.}
Separation-pair reduction reaches three new graphs but loses eleven, including seven at the memory limit and four at the time limit. Coverage falls from 826 to 818 and PAR-2 rises from 649.4 to 680.8. These failures reveal a memory cost that solved-run summaries alone do not capture. The main configuration therefore leaves this reduction disabled. The per-instance runtime and memory records are provided in the supplementary artifact.

\subsection{Validation on Hamiltonian and Non-Hamiltonian Cubic Graphs}
\label{sec:cubic-results}

The FHCP Challenge Set contains known-Hamiltonian graphs, so it exercises verified cycle construction but not an UNSAT conclusion. We therefore evaluated the directed cubic branch on the labelled cubic census from the Flinders graph database. The experiment includes all 1,927 non-Hamiltonian graphs of orders 10--18. At each order, we drew a Hamiltonian sample equal in size to the biconnected non-Hamiltonian subset, giving 272 positive instances. Every graph received a five-second cutoff.

Across the 2,199 graphs in Table~\ref{tab:cubic-results}, \method{} returns the correct decision for every instance. The articulation filter proves 1,655 of the 1,927 negative cases (85.9\%) before SAT solving. Directed cubic SAT--CEGAR decides the remaining 544 biconnected graphs, proving 272 cases UNSAT and constructing a verified cycle for all 272 matched positive cases.

Order 18 contains most of the nontrivial cases: after filtering, 231 biconnected negatives and 231 matched positives reach SAT--CEGAR. Correct decisions exercise both cycle construction and refutation. Together with all 88 cubic FHCP instances solved, this experiment validates the complete cubic branch on both positive and negative inputs.

\begin{table}[t]
  \caption{Cubic census by order. ``Artic. filter'' counts non-Hamiltonian graphs proved immediately by an articulation point; ``Biconn.'' counts the remainder passed to directed SAT--CEGAR. The Hamiltonian sample has equal size.}
  \label{tab:cubic-results}
  \centering\small
  \begin{tabular*}{\columnwidth}{@{\extracolsep{\fill}}lrrrr@{}}\toprule
    $|V|$ & Non-Ham. & Artic. filter & Biconn. & Ham. sample\\ \midrule
    10 & 2 & 1 & 1 & 1\\
    12 & 5 & 4 & 1 & 1\\
    14 & 35 & 29 & 6 & 6\\
    16 & 219 & 186 & 33 & 33\\
    18 & 1,666 & 1,435 & 231 & 231\\ \midrule
    Total & 1,927 & \textbf{1,655} & 272 & 272\\
    \bottomrule
  \end{tabular*}
\end{table}

\section{Discussion}
\label{sec:discussion}

The ablation identifies formulation selection as the main source of improvement on FHCP. High-degree-cluster and cubic graphs account for 81 of the 82 additional solves over the contracted general encoding. Both formulations replace many edge-level choices with a smaller set of global decisions: connections between clusters in one case, and one predecessor and successor per vertex in the other. Their lower CEGAR iteration counts support the view that a suitable representation can expose connectivity earlier in the search.

This result complements recent SAT-based HCP solvers that strengthen connectivity reasoning while encoding the input graph directly~\cite{ohashi2025cutset,fazekas2024userpropagators,funder2026hcup}. \method{} instead changes the search representation before refinement. Although HCup and Cut-set CEGAR solve more FHCP instances overall, \method{} contributes 34 cycles outside their union, including all 22 dense common misses. The three approaches therefore cover different parts of the benchmark.

Not every structural abstraction is useful. Paired-chain selection leaves coverage unchanged, degree-two contraction gives only a small net gain, and separation-pair reduction loses eight instances overall. The last result reflects the cost of additional local path searches and motivates disabling this reduction by default. More selective recognizers are needed before these components can contribute consistently.

The evaluation has three main boundaries. The recognizer thresholds were developed on FHCP, the ablations use one solver seed, and the non-Hamiltonian cubic census contains only small graphs. The present evidence therefore supports conclusions about FHCP coverage, while broader selector performance and difficult UNSAT behavior remain open. Evaluation on unseen benchmarks and multiple seeds is the next step. The 17 graphs missed by all three FHCP solvers are sparse and non-cubic, providing a concrete target for extending the representation portfolio.

\section{Conclusion}
\label{sec:conclusion}

\method{} is an exact solver that chooses a SAT representation from graph structure. It uses specialized formulations for recurring graph families and a complete general SAT--CEGAR fallback. Any cycle found after graph reduction is expanded and checked in the original input graph. On FHCP, most of the improvement from formulation selection comes from the high-degree-cluster and cubic formulations. \method{} also solves graphs missed by HCup and Cut-set CEGAR, showing that the methods have complementary strengths. The cubic census further confirms that the complete cubic branch handles both Hamiltonian and non-Hamiltonian graphs correctly. Together, these results show that choosing the representation before refinement can broaden the reach of exact SAT solving on structured HCP instances.

\appendix
\section{Coverage beyond Published SAT Results}
\label{app:residual-cohort}

The union of the per-instance SAT results reported by Ohashi et al.~\cite{ohashi2025cutset} and Funder~\cite{funder2026hcup} covers 972 of the 1,001 FHCP instances. Under the density groups in Figure~\ref{fig:outcome-profile}, the remaining 29 comprise 23 sparse and six dense graphs. \method{} finds verified Hamiltonian cycles for all six dense graphs and six of the sparse graphs.

Treewidth dynamic programming already covers the six dense graphs. \method{} also solves these six cases and adds six sparse instances beyond all published exact per-instance result sets available for FHCP, increasing their union from 978 to 984 instances. The other 17 instances are all sparse, with average degree between 3.16 and 3.42. Thus, within this residual cohort, every unresolved instance is sparse and non-cubic.

\begin{table}[H]
  \caption{Per-instance outcomes for the 29 FHCP graphs outside the published SAT union. Runtime and peak memory are reported for verified SAT outcomes.}
  \label{tab:residual-cohort}
  \centering\scriptsize
  \setlength{\tabcolsep}{2.25pt}
  \renewcommand{\arraystretch}{0.88}
  \begin{tabular*}{\columnwidth}{@{\extracolsep{\fill}}rrrrrrr@{}}\toprule
    ID & $|V|$ & $\bar d$ & Prior exact & Outcome & $t$ (s) & $M_{\mathrm{peak}}$ (MB)\\ \midrule
    710 & 4,064 & 3.35 & -- & \textbf{SAT} & 232.79 & 75\\
    717 & 4,122 & 3.71 & -- & \textbf{SAT} & 174.22 & 57\\
    746 & 4,286 & 8.53 & TW-DP & \textbf{SAT} & 50.06 & 114\\
    788 & 4,620 & 3.27 & -- & TO & -- & --\\
    832 & 5,070 & 3.15 & -- & \textbf{SAT} & 407.66 & 148\\
    868 & 5,544 & 3.27 & -- & TO & -- & --\\
    882 & 5,686 & 3.27 & -- & \textbf{SAT} & 30.27 & 80\\
    937 & 6,412 & 3.36 & -- & TO & -- & --\\
    944 & 6,544 & 3.34 & -- & MO & -- & --\\
    950 & 6,620 & 8.68 & TW-DP & \textbf{SAT} & 61.11 & 103\\
    951 & 6,630 & 3.19 & -- & TO & -- & --\\
    954 & 6,735 & 3.42 & -- & TO & -- & --\\
    959 & 6,925 & 3.39 & -- & \textbf{SAT} & 535.92 & 358\\
    960 & 6,930 & 3.27 & -- & TO & -- & --\\
    963 & 7,020 & 8.69 & TW-DP & \textbf{SAT} & 137.27 & 107\\
    965 & 7,102 & 3.29 & -- & TO & -- & --\\
    966 & 7,104 & 3.36 & -- & TO & -- & --\\
    971 & 7,295 & 3.26 & -- & \textbf{SAT} & 338.23 & 240\\
    974 & 7,418 & 3.20 & -- & TO & -- & --\\
    975 & 7,420 & 8.71 & TW-DP & \textbf{SAT} & 356.97 & 171\\
    976 & 7,434 & 3.20 & -- & TO & -- & --\\
    981 & 7,620 & 3.22 & -- & TO & -- & --\\
    982 & 7,620 & 8.72 & TW-DP & \textbf{SAT} & 280.25 & 166\\
    983 & 7,650 & 3.20 & -- & TO & -- & --\\
    987 & 7,850 & 3.16 & -- & TO & -- & --\\
    990 & 8,020 & 8.73 & TW-DP & \textbf{SAT} & 260.55 & 193\\
    993 & 8,380 & 3.18 & -- & TO & -- & --\\
    994 & 8,401 & 3.29 & -- & TO & -- & --\\
    998 & 8,613 & 3.33 & -- & TO & -- & --\\
    \bottomrule
  \end{tabular*}
  \par\smallskip
  \parbox{\columnwidth}{\scriptsize\emph{Notes.} ID $k$ denotes graph $k$, and $\bar d$ is average degree. Sparse graphs have $\bar d<4$; dense graphs have $\bar d\geq6$. ``Prior exact'' records additional published exact per-instance coverage: TW-DP denotes treewidth dynamic programming, and -- means that no prior exact result is reported. TO and MO denote timeout and memory-out.}
\end{table}

\balance
\bibliographystyle{ACM-Reference-Format}
\bibliography{references}
\end{document}